\documentclass[reqno,xcolor=dvipsnames]{amsart}
\usepackage[dvipsnames]{xcolor}
\usepackage{geometry}

\usepackage[utf8x]{inputenc}
\usepackage{amsmath}
\usepackage{graphicx}
\usepackage{rotating}
\usepackage{wasysym}
\DeclareGraphicsExtensions{.pdf,.png,.jpg}

\usepackage{geometry}
\usepackage[round,comma]{natbib}                %

\usepackage{enumerate}

\usepackage{amsmath}
\usepackage{amsthm}
\usepackage{amsfonts}
\usepackage{amssymb}
\usepackage{dsfont}

\usepackage{nicefrac}
\usepackage{booktabs}
\usepackage{mathrsfs}
\usepackage{bm} 
\usepackage{mathtools}                          %

\newcommand{\ccF}{{\mathscr F}}
\newcommand{\ccG}{{\mathscr G}}
\newcommand{\ccH}{{\mathscr H}}

\newcommand{\cV}{{\mathcal V}}

\newcommand{\Ind}{{\mathds 1}}
\newcommand{\ind}[1]{\Ind_{\{#1\}}}

\newcommand{\R}{\mathbb{R}}

\newcommand{\bbF}{\mathbb{F}}

\newcommand{\bbG}{\mathbb{G}}

\newtheorem{theorem}{Theorem}[section]
\newtheorem{corollary}[theorem]{Corollary}      %
\newtheorem{lemma}[theorem]{Lemma}              %
\newtheorem{proposition}[theorem]{Proposition}  %

\theoremstyle{definition}
\newtheorem{definition}[theorem]{Definition} %
\newtheorem{remark}[theorem]{Remark}%
\newtheorem{assumption}[theorem]{Assumption}%

\numberwithin{equation}{section}

\newcommand{\QcirP}{Q\mkern-2mu\odot\mkern-2mu P}

\DeclareMathOperator{\image}{im}

\usepackage{natbib}
\usepackage[colorlinks,urlcolor=red,citecolor=blue,linkcolor=red]{hyperref}

\usepackage{stmaryrd}

\DeclareMathOperator{\loc}{loc}

\begin{document}

\title[Benchmark-Neutral Risk-Minimization]{Benchmark-Neutral Risk-Minimization \\
for insurance products and nonreplicable claims}

		\author{Michael Schmutz}
		\address{Head of Financial Risks, Swiss Financial Market Supervisory Authority (FINMA), Laupenstrasse 27, 3003 Bern, and \'{E}cole Polytechnique F\'{e}d\'{e}rale de Lausanne (EPFL), Station 8, 1015, Lausanne, Switzerland. \newline \emph{This publication reflects solely the personal opinion of the author and is not binding for FINMA.}}
             \email{m.schmutz@epfl.ch}
		\author{Eckhard Platen}
		\address{University of Technology Sydney,  School of Mathematical and Physical Sciences, and  Finance Discipline Group, Australian National University, Canberra, College of Business and Economics}
		\email{eckhard.platen@uts.edu.au}

		\author{Thorsten Schmidt}
		\address{Albert-Ludwigs University of Freiburg, Ernst-Zermelo Str. 1, 79104 Freiburg, Germany.}
    \email{thorsten.schmidt@stochastik.uni-freiburg.de}

\thanks{Acknowledgements:\/ The authors would like to express their gratitude for receiving valuable suggestions on the  paper by Stefan Tappe and Gerhard Stahl. }

\maketitle

\begin{abstract}
In this paper we study the pricing and hedging of nonreplicable contingent claims, such as long-term insurance contracts like variable annuities. Our approach is based  on the benchmark-neutral pricing framework of \cite{Platen24}, which differs from the classical benchmark approach by using the stock growth optimal portfolio as the num\'eraire. In typical settings, this choice leads to an equivalent martingale measure, the \emph{benchmark-neutral} measure. The resulting prices  can be  significantly lower  than the respective risk-neutral  ones, making this approach  attractive for long-term %
risk-management. 
We derive the associated risk-minimizing hedging strategy under the assumption that the contingent claim possesses a martingale decomposition.  For a set of  nonreplicable contingent claims, these  strategies allow  monitoring  the working capital  required to generate their  payoffs and enable an assessment of  the resulting diversification effects. Furthermore,  an algorithmic refinancing strategy is proposed that allows   modeling   the working capital. Finally,  insurance-finance arbitrages of the first kind are introduced and it is demonstrated that benchmark-neutral pricing effectively avoids such arbitrages.
\end{abstract}




	\section{Introduction}\label{section.intro}

The pricing and hedging of nonreplicable long-term contingent claims has been a challenging task in finance and insurance. 
Portfolio strategies that aim to replicate such  contingent claims generate, in general,
fluctuating profit and loss processes. A major step  toward a satisfactory solution to this problem was made by the F\"ollmer-Sondermann-Schweizer local risk-minimization approach, originally proposed in  \cite{foellmer-sondermann-86}, which became further developed in \cite{FollmerSc91},  \cite{Schweizer91},  \cite{Schweizer01}, and \cite{FreySchmidt2012}. It minimizes the fluctuations of savings account-denominated profit and loss processes  using a quadratic criterion under an assumed minimal equivalent martingale measure. For a given contingent claim, the profit and loss process monitors in units of the savings account the adapted inflow
and outflow of capital to and from the hedge portfolio and forms a local martingale that  is orthogonal to savings account-denominated traded securities under the assumed  minimal equivalent martingale measure. The local risk-minimization approach provides an intuitively
appealing methodology for the pricing and hedging of nonreplicable contingent claims and  %
was introduced into the actuarial literature by \cite{Moller98} and \cite{Moller01}. Other approaches considering the pricing of non-traded quantities rely on a certain pasting of the risk-neutral measure with the objective or real-world probability measure; see for example \cite{dybvig1992hedging} and \cite{ArtznerEiseleSchmidt2024} and references therein. Since in our setting, a risk-neutral measure might not exist, we will focus on the weaker notion of arbitrages of the first kind; see \cite{kardaras2012market} and \cite{karatzas2021portfolio} for a detailed exposition.   

More specifically, the paper \cite{Platen24} shows that when modeling realistically the long-term  dynamics of well-diversified stock portfolios  it is unlikely that an equivalent risk-neutral probability measure 
  exists. If one ignores this fact and prices under local risk-minimization using an assumed putative minimal equivalent  martingale measure, then most  life insurance and pension contracts turn out to be   more expensive than necessary.   To achieve the most economical   hedges of  long-term contingent claims,  risk-neutral pricing and  local risk-minimization need to be replaced by a pricing method that leads to the  minimal possible prices. 
  
By applying the benchmark approach and its 
 pricing under the real-world probability measure, see \cite{PlatenHe06},  the method of local risk-minimization was modified in \cite{BiaginiCrPl14}.  This modified pricing method employs  the growth optimal portfolio (GOP)   of the entire market as num\'eraire  and the real-world probability measure as pricing measure. Since restrictive second-moment conditions with respect to primary security accounts   remain still present when employing this  risk-minimization method,  \cite{DuPl16} generalized this method to the more general  method of benchmarked risk-minimization. For a given contingent claim, the latter method is requesting  the GOP-denominated profit and loss process  to be orthogonal to the GOP-denominated traded securities, and provides via its real-world pricing formula the minimal possible price.
 
Although benchmarked risk-minimization addresses several theoretical challenges, its practical implementation  remains still  a challenge. Specifically, it requires the ability  to liquidly trade the GOP of the  entire market. This GOP is a highly leveraged portfolio that goes long in the stock GOP (the GOP of the stocks without the savings account) and   short in the savings account; see \cite{FilipovicPl09}. However in practice, the GOP of the entire market is not a suitable  num\'eraire, as it can only be approximated   through a  leveraged portfolio, which, due to  discrete-time trading,  can  potentially become   negative. %
Moreover, constructing this GOP necessitates estimating the expected return of the stock GOP -- a notoriously difficult task.

A widely applicable pricing and risk-minimization method must  rely on  a  strictly positive,  liquidly tradable num\'eraire because such a num\'eraire   is essential when hedging contingent claims. For this important practical reason, the present paper proposes the concept of {\em benchmark-neutral (BN) risk-minimization},  
where the stock GOP  is employed as a num\'eraire. The stock GOP serves as a viable num\'eraire since it can be  approximated  by a guaranteed strictly positive, tradable stock portfolio; see \cite{PlatenRe20}.	The new concept of BN risk-minimization builds on the  results in \cite{Platen24}, where  BN pricing  of  replicable contingent claims was first introduced. %

 For the pricing and hedging of nonreplicable contingent claims, we introduce the new method of BN risk-minimization. This method  uses  the stock GOP  as a num\'eraire and the BN pricing measure for pricing. Most importantly, it monitors for a given stock GOP-denominated nonreplicable contingent claim the nonhedgeable part of its stock GOP-denominated price process in a manner that makes this process under the BN pricing measure a local martingale and orthogonal to all stock GOP-denominated traded securities. Orthogonality means here  that the products of the stock GOP-denominated nonhedgeable part and stock GOP-denominated traded securities are local martingales under the BN pricing measure. We also show that such a pricing rule avoids the introduction of insurance-finance arbitrages of the first kind.

	The paper studies  the book of liabilities of  a line of business (LOB) of  an insurance company or a pension fund. This book  consists of a  set  of not fully replicable contingent claims and we apply BN risk-minimization to such a set of contingent claims by  involving the notions of  production portfolio, working capital, and refinancing sequence. We  show how the diversification effect reduces the average working capital needed per contingent claim when managing the remaining randomness and indicate a way  to regulate and supervise the production of the payoffs of the contingent claims in a  book of liabilities of an LOB.

	      The paper is organized as follows: Section 2 introduces the  benchmark approach with its real-world pricing formula. The concept of BN pricing is described in Section 3. BN hedging is introduced in Section 4. The new method of BN risk-minimization is proposed in Section 5. Section 6 applies this method to the management of the  working capital of an LOB. Finally, Section 7 studies insurance-finance arbitrages and shows that BN pricing does not permit insurance-finance arbitrages of the first kind.

\section{The financial market}

The current paper  models  financial  markets under the usual frictionless assumptions, i.e.\ allowing for  continuous trading,    instantaneous investing and borrowing,   short sales with full use of proceeds, and   infinitely divisible securities. For simplicity, we assume no transaction costs.

\subsection*{Primary security accounts}

Let us consider a filtered probability space $(\Omega,\ccF,\bbF,P)$ satisfying the usual conditions, see, e.g., \cite{JacodShiryaev}. The filtration $\bbF=(\ccF_t)_{t \ge 0}$  describes the evolution of the market information
	 over time. 	 
The market consists of $n\in \{1,2,\dots\}$ \emph{primary security accounts}, such as  stocks, where all  dividends or payments are reinvested. Additionally, there exists a riskless \emph{ savings account} $S^0$. 
The primary security accounts, dominated in units of the  savings account, are denoted by $S^1,\dots,S^n$ and we assume that these are adapted,  nonnegative diffusions.  Additionally, we have $S^0\equiv 1$. %

Moreover, we assume that $S^1,\dots,S^n$ are driven by an $n$-dimensional standard $\bbF$-Brownian motion $W=(W^1,\dots,W^n)^\top$. It should be pointed out, that this does not imply completeness of the market, since $\bbF$ is not necessarily generated by $W$.

\begin{remark}
	Within this paper we denominate, for simplicity, in units of the savings account, which avoids the modeling of the interest rate. Of course, it would  be possible to denominate in units of the currency and quantify the resulting impact of stochastic interest rates. Moreover, it is  possible to generalize the setting to semimartingales, see, e.g., \cite{kardaras2012market} and \cite{DuPl16}.
\end{remark}
	 
\subsection*{The stock GOP}	 
We assume  for the given primary security accounts that there exists a continuous \emph{growth optimal portfolio} (GOP), which we denote by $S^*$, again denominated in units of the  savings account. We call this GOP the \emph{stock GOP}. Note that the investment universe of the stock GOP does not include the savings account $S^0$. The primary security accounts denominated in the stock GOP are denoted by %
\begin{align}\label{eq:X}
	{X}^j= %
	\frac{S^j}{S^*}
	,\qquad j =0, 1,\dots,n.
\end{align}
Under minimal assumptions, the securities $X^1,...,X^n$ are $\bbF$-$P$-local martingales (see \cite{FilipovicPl09}, \cite{kardaras2012market}). In particular, the  no arbitrage of the first kind (NA${}_1$) condition is for the market formed by $S^1,...,S^n$ equivalent to the existence of a local martingale deflator, such as $S^*$. 

Under the assumption that $S^*$ exists, it follows by  Theorem 3.1 of \cite{FilipovicPl09} that the unique stock GOP $S^*$  emerges in the form
	\begin{align}
		\label{e.4'}
	\frac{dS^{*}_t}{S^{*}_t}=\lambda^*_t \, dt + \sigma^*_t{}^\top \cdot  \, (  \sigma^*_t \, dt+d W_t), \quad t \ge 0
	\end{align}  
	  with fixed initial value $S^*_0>0$.
Here, $\lambda^*$ denotes the {\em net risk-adjusted return} of  $S^*$  and $\sigma^*$ the {\em volatility} of $S^*$. Both parameter processes can be obtained from the dynamics of $S$; see \cite{FilipovicPl09} for details. 

\smallskip

\subsection{The classical GOP}\label{sec:classical GOP}
 Let us  extend the above market, formed by the $n$ primary security accounts,  by adding the riskless savings account $S^0$ as an additional primary security account. We will show below, that the (unique) \emph{GOP of the extended market} always exists, and denote  it by    $S^{**}$. The use of the extended market GOP  $S^{**}$ as GOP is  in line with the classical benchmark approach.  In practical applications, however, it turns out that $S^{**}$ is a highly leveraged portfolio that goes long in the stock GOP and short in the savings account. This might cause  difficulties when trading  in discrete time because one might not be able to guarantee strict positivity for a tradable approximation of $S^{**}_t$. The guaranteed strict positivity is not a problem for the stock GOP $S^*_t$, which  can be approximated by a guaranteed strictly positive  total return stock index, as shown in \cite{PlatenRe20}.

The aim of the next result is to establish a precise connection to the classical GOP of the benchmark approach, namely the GOP $S^{**}$ of the extended market. %
To this end, we  add the risk-less savings account $S^0=1$ to the market formed by the $n$ risky primary security accounts $S^1,...,S^n$. %
 We  need to introduce a precise description of this market and 
 assume according to Theorem 3.1 in \cite{FilipovicPl09} the return process of the vector $\tilde S=(S^1,\dots,S^n)^\top$ to satisfy
	\begin{align}\label{eq: dS}
		\frac{d\tilde S_t}{\tilde S_t}= a_t \, dt + b_t \, dW_t=\lambda^*_t{\bf {1}}dt+b_t(\sigma^*_tdt+dW_t)=b_t(\sigma^{**}_tdt+dW_t) , \quad  t \ge 0,  
	\end{align}
with
\begin{equation}\label{eq:sigma **'} 
\sigma^{**}_t=\lambda^*_tb^{-1}_t{\bf {1}}+\sigma^*_t,\end{equation}
	where we write $\nicefrac{d\tilde S_t}{\tilde S_t}$ for the vector of stochastic differentials $\nicefrac{d\tilde S^j_t}{\tilde S^j_t}$, $j=1,\dots,n$, 
	with an adapted and locally integrable vector process  $a$, taking values in $\R^n$,  and an adapted,  locally square-integrable, non-singular matrix process $b$ taking values in $\R^{n \times n}$, respectively. 
	The vector of primary security accounts of the extended market (discounted in terms of the savings account) is denoted by $S=(S^0,\dots,S^n)^\top$ (in contrast to $\tilde S$ that is without the  savings account).

\begin{proposition}\label{prop:2.2}
	The market formed by $\tilde S=(S^1,\dots,S^n)^\top$ with dynamics given by \eqref{eq: dS} has the stock GOP $S^*$ as GOP satisfying \eqref{e.4'}. The market extended by the riskless savings account $S^0 \equiv 1$  admits the unique GOP $S^{**}$ satisfying
	\begin{align}\label{eq:S**}
		\frac{dS^{**}_t}{S^{**}_t} =  \sigma^{**}_t{}^\top   (\sigma_t^{**}  dt +  dW_t) , \quad t \ge 0.
	\end{align}
\end{proposition}
 
\begin{proof}
It is straightforward to verify the assumptions of Theorem 3.1 in \cite{FilipovicPl09} for the case of the market formed by the $n$ risky primary security accounts.
	To derive \eqref{e.4'}, define the $\R^{(n+1)\times(n+1)}$-valued process $M$ by
	$$ M_t = \left(\begin{matrix}
		b_t \cdot b_t^\top  & \mathbf{1}\\ \mathbf{1}^\top  & 0
	\end{matrix}\right), \qquad t \ge 0. $$
	Hence, except from a $dt\otimes dP$-nullset and because $b_t$ is  nonsingular, we have
	$$  {a_t \choose 1} \in \image (M_t) $$
and it follows by Lemma 3.1 and Theorem 3.1 in \cite{FilipovicPl09} \eqref{e.4'}, where the vector $(	\pi^*_t ,\lambda^*_t )^\top$ solves the matrix equation
	$$ M_t \left(\begin{matrix}
	\pi^*_t \\\lambda^*_t 
	\end{matrix}\right)=	  {a_t \choose 1}  $$\\
	yielding \begin{equation}\label{sigma*t}
	\sigma^*_t=b_t^\top \pi^*_t , \qquad\pi_t^{*}{}^\top {\bf {1}}=1, \qquad t \ge 0.\end{equation}
 Analogously, the GOP $S^{**}_t$ of the extended market  exists  by Theorem 3.1 in \cite{FilipovicPl09}. To see this, define the $\R^{(n+2)\times(n+2)}$-valued process $\tilde M$ by
 $$ \tilde M_t = \left(\begin{matrix}
0 & {\bf {0}}^\top  & 1\\ {\bf {0}} &  b_t \cdot b_t^\top & \mathbf{1} \\ 1 &\mathbf{1}^\top   & 0
 \end{matrix}\right), \qquad t \ge 0. $$
 Except from a $dt\otimes dP$-nullset, we have
 $$ \left(\begin{matrix} 0 \\ a_t \\ 1 \end{matrix}\right) \in \image (\tilde M_t) $$ 
 and it follows   by Theorem 3.1 in \cite{FilipovicPl09} \eqref{eq:S**} and \eqref{eq:sigma **'}, where the vector $( \pi^{**,0}_t,(\pi^{**}_t)^\top ,\lambda^{**}_t)^\top$ solves the matrix equation 
 	$$ \tilde M_t \left(\begin{matrix}
 	\pi^{**,0}_t\\ \pi^{**}_t \\\lambda^{**}_t 
 	\end{matrix}\right)=	  \left(\begin{matrix} 0 \\a_t \\ 1 \end{matrix}\right)=	  \left(\begin{matrix} 0 \\\lambda^*_t{\bf {1}}+b_t\sigma^*_t \\ 1 \end{matrix}\right) $$
 	with
 	$$\lambda^{**}_t=0 $$
 because the denominating savings account is a traded security,	yielding
 	\begin{equation} \label{sigma**}
 	\sigma^{**}_t=b^\top_t \pi^{**}_t=\lambda^*_tb^{-1}_t{\bf {1}}+\sigma^*_t, \qquad t \ge 0.\end{equation}
\end{proof}

\subsection{The simplified market setting}\label{sec:simplified}
For illustration we discuss throughout the paper also a {\em simplified market setting}, where  we assume in \eqref{e.4'} only one driving Brownian motion such that the stock GOP satisfies
	\begin{align}
		\label{e.4 1d}
	\frac{dS^{*}_t}{S^{*}_t}=\lambda^*_t \, dt + \theta_t\,  (  \theta_t \, dt+d B_t), \quad t \ge 0,
	\end{align}
with \begin{equation}
\label{thetasigma*}
\theta = ||\sigma^* ||,\end{equation}
where $B= ((\theta_t)^{-1} \sigma_t^*{}^\top W_t)_{t \ge 0}$ is a one-dimensional Brownian motion. %
  For illustration, in Section \ref{sec:MMM} the  stock GOP dynamics of this simplified  market setting will be specified.\\
 
We can extend this simplified market setting by adding the riskless savings account 
$S^0 \equiv 1$, which leads to the following conclusion: 
\begin{corollary}
	For the extended simplified market setting with the stock GOP $S^*$ satisfying \eqref{e.4 1d} with $\theta>0$ $dt \otimes dP$-a.s.~and the riskless savings account $S^0 \equiv 1$,   the unique  GOP $S^{**}$ satisfies
	\begin{align}\label{eq:GOP 1d}
		\frac{dS^{**}_t}{S^{**}_t} =  \sigma^{**}_t (\sigma_t^{**} dt +  dB_t) 
	\end{align}
	with
	\begin{align}\label{eq:sigma **2}
		\sigma_t^{**} =  \frac{\lambda_t^* }{\theta_t}+\theta_t.
	\end{align}
\end{corollary}
\begin{proof}
	In the case where $n=1$,
	 \eqref{e.4 1d} implies that $b_t = \theta_t=||\sigma^*_t||$ in Equation \eqref{eq: dS}.
	Then, the result follows by applying Proposition \ref{prop:2.2}.
\end{proof}

\subsection{Real-world pricing}
Pricing under the classical benchmark approach is called \emph{real-world pricing} because in contrast to risk-neutral pricing it utilizes expectations under the statistical or real-world probability measure $P$, which is done as follows: 
	Consider a bounded stopping time $T > 0$, and let $L^1_Q(\ccF_T)$
	denote the set of integrable $\ccF_T$-measurable random variables with respect to a probability measure $Q$.
A \emph{contingent claim} is equivalently described by  a payoff $H_T\geq 0$ with $\frac{H_T}{S^{**}_T}\in
	 L^1_P(\ccF_T)$. 

Following the notions defined in \cite{PlatenHe06}, %
we call a security price process $H$ \emph{fair} (under real-world pricing), if its benchmarked price process $\hat H = \frac{H}{S^{**}}$ is an $\bbF$-$P$-martingale. From this, we directly obtain the \emph{real-world pricing formula} 
\begin{align}\label{eq:real-world pricing}
	H_t=S^{**}_t \cdot E_P \left[\frac{H_T}{S^{**}_T} \, \Big|\, \ccF_t\right], \qquad 0 \le t \le T,
\end{align}
where $ E_P [\,  \cdot  \, |\, \ccF_t]$ denotes the conditional expectation with respect to $P$ under the information available at the time $t\in[0,T]$.	The real-world pricing formula uses the GOP $S^{**}$ of the extended market (consisting of the primary security accounts, such as stocks, together with the riskless savings account) as the num\'eraire and the real-world probability measure $P$ as the pricing measure. 
	
	When the contingent claim is replicable, it is  shown in \cite{DuPl16} that the fair price process coincides with the minimal possible self-financing hedge portfolio   that hedges its payoff.
	For the case when the contingent claim is not replicable,   \cite{BiaginiCrPl14} and \cite{DuPl16} proposed  risk-minimization methods that employ the GOP of the extended market $S^{**}$ as the num\'eraire and the real-world probability measure $P$ as the pricing measure.\\
		 For not replicable contingent claims, there exist   other   risk-minimization approaches in the literature, including the previously mentioned local risk-minimization approach; see \cite{FollmerSc91} and \cite{Moller98}. 
		  However,  the respective  pricing rules do never provide lower prices for nonnegative replicable contingent claims  in comparison to the real-world pricing formula because all $S^{**}$-denominated self-financing portfolios, which replicate a given nonnegative contingent claim, are $\bbF$-$P$-supermartingales and the $\bbF$-$P$-martingale is the least expensive one; see Lemma A.1 in \cite{DuPl16}.

\section{Benchmark-neutral pricing}
	 
Following 	\cite{Platen24}, we consider \emph{benchmark-neutral pricing} as a less expensive alternative to risk-neutral pricing and as a more robust alternative to real-world pricing. The key of this approach is to employ the stock GOP $S^*$ as a num\'eraire instead of the GOP $S^{**}$ of the extended market.

	 To obtain pricing formulae in the spirit of the classical risk-neutral pricing methodology, we  introduce the Radon-Nikodym derivative $\Lambda$ as the normalized ratio of the stock GOP $S^*$ over the GOP $S^{**}$ of the extended market %
	 $$  \Lambda := \frac{\frac{S^*}{S^*_0}}{\frac{S^{**}}{S^{**}_0}}. $$

\begin{lemma}\label{sigma*minus}
	The dynamics of the Radon-Nikodym derivative process $\Lambda$ is given by
	$$ \frac{d\Lambda_t}{\Lambda_t} = - \mu^*_t dW_t $$
	with $$ \mu^*_t=(\sigma^{**}_t - \sigma^*_t)^\top=\lambda^*_t{\bf {1}}^\top(b^{-1}_t)^\top , \qquad t \ge 0.$$
\end{lemma}

\begin{proof}
	Applying the It\^o formula and using the dynamics of $S^*$ in \eqref{e.4'} and the dynamics of $S^{**}$ in \eqref{eq:S**}, we obtain 
	\begin{align*}
		\frac{d\Lambda_t}{\Lambda_t} &= \Big( \lambda_t^* + \sigma_t^{*}{}^\top \sigma_t^* - \sigma_t^*{}^\top \cdot \sigma_t^{**} \Big) dt + (\sigma_t^* - \sigma_t^{**})^\top dW_t.
	\end{align*}
By Equation \eqref{sigma**}	it follows
\begin{equation}\label{sigma**-} \sigma_t^* - \sigma_t^{**}= -\lambda^*_tb^{-1}_t{\bf {1}},
	\end{equation}
	which yields with \eqref{sigma*t} the identity
	\begin{equation}
	\label{lambda*0}
	\lambda_t^* + \sigma_t^{*}{}^\top \sigma_t^* - \sigma_t^*{}^\top \cdot \sigma_t^{**} =\lambda_t^*(1 - \sigma_t^{*}{}^\top b^{-1}_t{\bf {1}})=\lambda_t^*(1 - \pi_t^{*}{}^\top b_t b^{-1}_t{\bf {1}})=\lambda_t^*(1 - \pi_t^{*}{}^\top {\bf {1}})=0 \end{equation}
	 resulting in
	\begin{align*}
	\frac{d\Lambda_t}{\Lambda_t} =  - \lambda^*_t(b^{-1}_t{\bf {1}})^\top dW_t.
\end{align*}
\end{proof}
\begin{remark}[Simplified market setting] \label{rem:3.3}
	For the simplified market setting, introduced in Section \ref{sec:simplified}, we can directly compute
	$$ \mu^*_t=\sigma_t^{**} - \sigma_t^{*} = \frac{\lambda_t^*}{\theta_t}, $$
	by Equation \eqref{sigma**-}, \eqref{thetasigma*}, and \eqref{eq:sigma **2};
	see also Equation (3.3) in \cite{Platen24}.	
\end{remark}
As can be seen, the Radon-Nikodym derivative $\Lambda$ is an $\bbF$-$P$-local martingale. This also follows from the fact that $\Lambda$ can be interpreted as a self-financing portfolio that is denominated in the GOP $S^{**}$ of the extended market.
 Under suitable dynamics of $ \Lambda$, the real-world pricing formula \eqref{eq:real-world pricing} indicates that $ \Lambda$ could be an $\bbF$-$P$-martingale. This leads us to the following assumption, which we will employ throughout the remainder of the paper:
 \begin{assumption}\label{ass}
 	The  Radon-Nikodym derivative process $ \Lambda$ is an  $\bbF$-$P$-martingale.
 \end{assumption}
 
 Under this assumption it is  possible to work  with a finite time horizon $T^*\in(0,\infty)$ and %
  proceed with a num\'eraire  change, as described in \cite{GemanKarouiRochet95}, by introducing the equivalent \emph{benchmark-neutral pricing measure} $Q^*$ satisfying
   $$dQ^* = \Lambda_{T^*} dP.$$ 
	\begin{remark}[Alternative assumption]
		We could alternatively assume that the Radon-Nikodym derivative process $ \Lambda$ is a uniformly integrable  $\bbF$-$P$-martingale, which is a slightly stronger assumption because it requires additionally  that $\Lambda$ satisfies  $\lim_{K\rightarrow \infty}\sup_{t \ge 0}E_P[|\Lambda_t|\ind{\Lambda_t> K}]=0$.  Hence the $P$-almost sure limit $\lim_{t\rightarrow \infty} \Lambda_t=\Lambda_\infty$ would be well-defined and one could    introduce the equivalent BN pricing measure $Q^*$ by setting
		\begin{align}
		dQ^* = \Lambda_\infty dP,
		\end{align} 
		which would allow  employing $Q^*$ over the full time interval $[0,\infty]$.
	\end{remark}

\subsection{The BN pricing formula}
Following  \cite{Platen24},  we  introduce the \emph{benchmark-neutral (BN) pricing formula}: 
Recall that $L^1_{Q^*}(\ccF_T)$ denotes the set of $\ccF_T$-measurable and $Q^*$-integrable random variables. Then, for  a contingent claim $H_T \ge 0$ with a bounded stopping time $T$ and $H_T (S_T^*)^{-1} \in L^1_{Q^*}(\ccF_T)$ 
we introduce the process $H^*$ by
\begin{align} 
		\label{BPF}H^*_t :=S^*_t \, E^* \left[\frac{H_{T}}{S^*_T} \, \Big| \,\ccF_t\right], \qquad 0 \le t \le T,  
\end{align}
Where $E^* [ \, \cdot \, | \,\ccF_t]$ denotes the conditional expectation with respect to $Q^*$ under the information available at time $t\in[0,T]$. Adding to the extended market the derivative with price process $H^*$ does not change the GOP $S^{**}$ of the extended market formed by the $n$ primary security accounts and the  savings account, as should be expected.  Key properties of BN pricing are summarized in the following theorem: 
\begin{theorem}\label{thm:BNpricing}
Under Assumption \ref{ass}, the following holds:
\begin{enumerate}[(i)]
	\item Let $H_T (S_T^*)^{-1}\in L^1_{Q^*} (\ccF_T)$ and $H_T (S_T^{**})^{-1}\in L^1_P(\ccF_T)$ with bounded stopping time $T\le T^*$. Then, its fair price $H$ satisfies
	\begin{align}\label{eq:H H*}
		H = H^*.
	\end{align} 
	\item The process 
	 \begin{equation}\label{eq:W*}
	 W^*= \int_0^\cdot  \mu^*_s \, ds +W
	 \end{equation}
	 with $\mu^*=(\sigma^{**} - \sigma^*)^\top$ is a $Q^*$-Brownian motion and 
	 \begin{align}\label{eq:stock GOP dynamicsy}
	 	\frac{dS^{*}_t}{S^{*}_t} &=\sigma^*_t{}^\top \cdot ( \sigma^*_t \, dt  + dW^*_t), \qquad t \ge 0.
	 \end{align}
	 \item $X$ is an $\bbF$-$P$-local martingale if and only if $X\Lambda^{-1}$ is an $\bbF$-$Q^*$-local martingale.
	 \item With $S^*$ as num\'eraire, $Q^*$ is an equivalent local martingale measure for $S=(S^0,S^1,\dots,S^d)^\top$. 
\end{enumerate}
\end{theorem}
\begin{proof}
	The first result follows as in \cite{GemanKarouiRochet95} by an application of Bayes' rule. Indeed:
	\begin{align}
		\frac{H^*_t}{S_t^*} = \,  E^* \left[\frac{H_{T}}{S^*_T} \, \Big| \,\ccF_t\right] = \, E \left[\frac{\Lambda_T}{\Lambda_t} \, \frac{H_{T}}{S^*_T} \, \Big| \,\ccF_t\right] = 
		 \, E \left[\frac{S_t^{**}}{S_t^*} \, \frac{H_{T}}{S_T^{**}} \, \Big| \,\ccF_t\right] = 
		 \frac{S_t^{**}}{S_t^*} \, \frac{H_t}{S_t^{**}},
	\end{align}
	and, since $S^*>0$ P-a.s, we obtain $H=H^*$. Next, by the  Girsanov theorem (see e.g.~Theorem III.3.11 in \cite{JacodShiryaev}),
	\begin{align}
		W^* %
		 = W + \int_0^\cdot \mu^* ds 
	\end{align} 
	is a $Q^*$-Brownian motion. 
For (ii), recall from Equation \eqref{e.4'} that
	\begin{align*}
	\frac{dS^{*}_t}{S^{*}_t} &=\lambda^*_t \, dt + \sigma^*_t{}^\top \cdot   (  \sigma^*_t \, dt+d W_t) \\
	&= \lambda^*_t+ \sigma^*_t{}^\top \cdot      \sigma^*_t - \sigma^*_t{}^\top \cdot \mu^*_t \, dt + \sigma^*_t{}^\top \cdot dW^*_t \\
	&= \lambda^*_t  + \sigma^*_t{}^\top \cdot ( \sigma^*_t - \sigma^{**}_t) + \sigma^*_t {}^\top \cdot \sigma^*_t dt  
	+ \sigma^*_t{}^\top \cdot dW^*_t \\
	&=  \sigma^*_t{}^\top \cdot ( \sigma^*_t  + dW^*_t),
	\end{align*}  
where for the last equality we used Equation \eqref{eq:sigma **'} as in \eqref{lambda*0}.

The third result follows  by the classical change of num\'eraire technique; see \cite{GemanKarouiRochet95}. For the convenience of the reader we provide the short proof: consider the non-negative $\bbF$-$P$-local martingale $X$ such that $dX = \sigma^\top  dW$ and 
note that by the It\^o formula
$$ \frac{d \Lambda^{-1}_t}{\Lambda^{-1}_t} = \mu^*_t dW_t + \mu^*_t{}^\top \mu^*_t dt,
$$
such that 
\begin{align*}
	dX_t \Lambda^{-1}_t &= (\sigma_t^\top  + \Lambda^{-1}_t\mu^*_t) \cdot dW_t + (\Lambda^{-1}_t\mu_t^*{}^\top + \sigma_t^\top) \cdot \mu_t^* \, dt   \\
	&= (\sigma_t^\top  + \Lambda^{-1}_t\mu^*_t) dW_t^*
\end{align*}
and hence, $X$ is an $\bbF$-$Q^*$-local martingale. The reverse direction follows in an analogous manner.

Regarding (iv), we simply note that since $S^{**}$ is a local martingale deflator for $S$, $S^*$ is a local martingale deflator for $S$ under $Q^*$ by (iii). 
\end{proof}
The above theorem shows two important properties: First, fair pricing, as in Equation \eqref{eq:real-world pricing}, and BN pricing, as in Equation \eqref{BPF}, coincide. Second, $Q^*$ is an equivalent local martingale measure for prices discounted in terms of the stock GOP $S^*$. As already mentioned, this implies that pricing and hedging under the BN methodology is, in practice, highly feasible. Note that, in general, the change to the savings account as the num\'eraire may not lead to an equivalent local martingale measure  such that classical risk-neutral pricing (as we will later  show in an example) would imply  strictly higher prices of securities than necessary; see Section \ref{sec:MMM} for a detailed example. In particular, for  long-term products in insurance and for pensions this fact plays a significant role. 
We refer to \cite{Platen24}  for further details on how to implement BN pricing and hedging for replicable contingent claims.

 A straightforward application of the above theorem is that the risk-less savings account $S^0$ when denominated with the stock GOP $S^*$, denoted by $X^0 = \nicefrac{S^0}{S^*}$, satisfies 
		\begin{align} 
		\label{e.4'''}
		dX^{0}_t=-X^0_t \, 
		 \sigma^*_t{}^\top  \cdot d W^{*}_t, \qquad t \ge 0,  
		\end{align} 
		with $X^0_0>0$, where $W^*$ is a $Q^*$-Brownian motion.

\section{Benchmark-neutral risk minimization}
\label{sec:Benchmark-neutral risk minimization}

The current paper  aims to identify the least expensive  method of delivering a targeted, potentially not fully replicable, contingent claim through hedging  while minimizing the fluctuations
 in the stock GOP-denominated hedging error under the BN pricing measure $Q^*$. 
For this aim, we first need to introduce the associated hedging strategies. 
If we consider the stochastic integral with respect to the continuous $\bbF$-$Q^*$-local martingale $X$, we denote by $L_{\loc,Q^*}^2(X)$ the class of progressively measurable locally-square integrable processes $\zeta$ such that 
$$ Q^*\Big( \int_0^t \zeta_s^i \zeta_s^j d \langle X^i, X^j \rangle_s < \infty\Big) = 1, 
$$ for all $t > 0$ and $0 \le i,j \le n$.

\subsection*{Admissible trading strategies}
The market participants can combine primary security accounts to form portfolios. 
As previously, denote by  
$
	X=\nicefrac{S}{S^*}	
$
  the assets denominated in units of the stock GOP $S^*$.

A trading strategy
$\nu=(\eta,\zeta)$ consists of an $\mathcal{F}_0$-measurable initial value $\zeta^\top_0 X_0$, an adapted process $\eta$,  and the $(n+1)$-dimensional process $\zeta \in L^2_{\loc}(X,Q^*)$. Here, $\zeta^j_t$, $j\in\{0,1,...,n\}$, represents the number of units of the $j$-th primary security account that are held at time $t \ge 0$ and $ X^\zeta = \zeta^\top \cdot X$ denotes the value of the portfolio. The portfolio trading strategy $\zeta$ is called \emph{self-financing}, if  
\begin{equation}\label{stochint}
X^\zeta=\zeta_0^\top X_0+\int_{0}^{\cdot }\zeta_t^\top \cdot dX_t.
\end{equation}
Intuitively, self-financing means that no  funds flow in or out of the portfolio. 

We call the trading strategy
$\nu=(\eta,\zeta)$ \emph{admissible}, if $\zeta$ is self-financing and 
\begin{align}\label{def:admissible-supermartingale}
	V^\nu = \zeta^\top \cdot X + \eta
\end{align} 
is a $Q^*$-supermartingale. We will  call also $\zeta$  a trading strategy by associating it with $\nu=(0,\zeta)$. 

Recall that by Theorem \ref{thm:BNpricing}, $X$ is an $\bbF$-$Q^*$-local martingale, and since $\zeta \in L^2_{\loc,Q^*}(X)$, then so is $X^\zeta$. 
If $\eta$ is predictable and of finite variation, then $V^\nu$ is a special semimartingale. When $\nu$ is admissible,  it is also an $\bbF$-$Q^*$-supermartingale, and hence $\eta$ is non-increasing. If $\eta$ is not of finite variation or only adapted but a special semimartingale, it can be decomposed into an $\bbF$-$Q^*$-local-martingale part, representing unhedgeable risks, and a non-increasing part of finite variation. We will derive useful properties of $\eta$ in Section \ref{sec:hedging}.

With this class of admissible dynamical trading strategies one can form a wide range of portfolios that include many not self-financing portfolios. 
The  admissible trading strategy $\nu$ employs
the stock GOP $S^*$ as the num\'eraire and monitors the inflow and outflow of extra capital by $\eta$ (again in units of the stock GOP  $S^*$).
For a replicable contingent claim, BN-hedging can be performed with a self-financing portfolio  as described in \cite{Platen24}. The current paper generalizes these results to the case of not replicable contingent claims.

We note that by Proposition \ref{prop:2.2}, see also Theorem 3.1 in \cite{FilipovicPl09}, the stock GOP $S^*$ is given by a self-financing trading strategy $\zeta^*$, such that $\zeta^*{}^\top S = S^*$ (and hence $\zeta^*{}^\top X=1$). 
When allowing
extra capital inflow or outflow by an admissible dynamic trading strategy, one obtains directly  the following
result:
\begin{corollary}\label{dynportf}
Consider an admissible dynamic trading strategy $\nu=(\eta,\zeta)$ and assume that $\eta \zeta^*{}^\top \in L^2_{\loc,Q^*}(X)$. Then, 
	 \begin{equation}\label{tildeV'}
	  V^\nu_t=\delta_t^\top X_t %
	 \end{equation}
	where $\delta \in L^2_{\loc,Q^*}(X)$ is given by 
	\begin{equation}
	\delta = \zeta +\eta \zeta^*{}^\top.
	\end{equation}
	\end{corollary}

 We introduce the following notion to capture the  payoff of a contingent claim. Recall that a contingent claim is equivalently described by its $\ccF_T$-measurable payoff $H_T \ge 0$ at the bounded stopping time $T$. 
 \begin{definition} \label{deliver}
 	An admissible trading strategy $\nu$ is said to {\em deliver} the  contingent claim $H_T$, if  
 	\begin{equation}
 	 V^\nu_T= \frac{H_T}{S_T^*} \qquad Q^*\text{-a.s.}
 	\end{equation}
 	 \end{definition}
 The contingent claim is \emph{replicable} if a self-financing trading strategy $\nu=(0,\zeta)$ exists such that $\zeta$ delivers the contingent claim.

\subsection{Hedging}\label{sec:hedging}
 	
 	  If one aims to deliver a not fully replicable contingent claim, one typically faces  a fluctuating stock GOP-denominated profit and loss %
 	   process and, thus, an intrinsic risk that the current paper aims to minimize. %
  	  	There exist many ways of hedging a not replicable contingent claim. One of these approaches is the  concept of local risk-minimization %
  	 by \cite{FollmerSc91}.  This approach employs the riskless savings account $S^0$ as the num\'eraire  and an assumed putative minimal equivalent martingale measure as the pricing measure. As follows from \cite{Platen24},  such an equivalent minimal martingale measure is unlikely to exist for realistic long-term stock GOP models and, as a consequence, the resulting prices of many long-term contingent claims become in this case significantly higher  than necessary; see also Section \ref{sec:putative RN measure}. 
  	 
  	Alternatively,  \cite{BiaginiCrPl14} and \cite{DuPl16} have used the GOP $S^{**}$ of the extended market as the num\'eraire and the real-world probability measure $P$ as the pricing measure under the benchmark approach. The respective risk-minimization methods provide the minimal possible prices for targeted contingent claims. As already mentioned, these methods face difficulties in practical implementations through the fact that the GOP of the extended market  $S^{**}$ is a highly leveraged portfolio.
  	  The concept of BN-pricing and the associated risk-minimization avoids this problem by using the stock GOP $S^*$ as a num\'eraire.  
  To introduce  BN risk-minimization, we follow \cite{Schweizer91}, but use $S^*$ as the num\'eraire. \\
  	 
  An  admissible  trading strategy $\nu=(\eta,\zeta)$  is called \emph{$Q^*$-mean self-financing} if $\eta$ is an  $\bbF$-$Q^*$-local martingale. 
  If $\eta$ is orthogonal to $X$ under $Q^*$, i.e.\ if $\eta \, X$ is an  $\bbF$-$Q^*$-local martingale, then we say that $\eta$ is a  {\em $Q^*$-orthogonal}  monitoring process. 
    We combine these two desirable  properties as follows:
    
     Consider a contingent claim $H_T \ge 0$ such that $\nicefrac{H_T}{S^*_T} \in L^1_{Q^*}(\ccF_T)$. Then we denote by
  $\cV(H_T) $  
  the set of admissible trading strategies which deliver $H_T$, are $Q^*$-mean self-financing, and have $Q^*$-orthogonal  monitoring processes.
 There may exist several such trading strategies for a contingent claim. In \cite{DuPl16}, the notion of a benchmarked risk-minimization strategy has been introduced that offers the cheapest value process in the sense of Equation \eqref{eq:V minimal} below. Following this idea, we introduce in the following definition the \emph{benchmark-neutral risk-minimizing strategy (BNRM)} that meets the following four desirable properties: first, it delivers the payoff; second, it is $Q^*$-mean self-financing; third, the monitoring process is $Q^*$-orthogonal; and, finally, it has a minimal value process in the sense of Equation \eqref{eq:V minimal} below.

   \begin{definition}\label{BPRISKMIN}
   	For  a contingent claim $H_T \ge 0$ with the bounded stopping time $T>0$ as maturity    and $\nicefrac{H_T}{S^*_T} \in L^1_{Q^*}(\ccF_T)$, a trading strategy $\nu \in \cV(H_T)$ is called a {\em BNRM strategy}, if 
    \begin{equation}\label{eq:V minimal}
   		 V^\nu_t\leq V^{\tilde \nu}_t, 
   	\end{equation}	
   	$Q^*$-almost surely	for all $t\in[0,T]$ and $\tilde \nu \in \cV(H_T)$.  
   	\end{definition}

   The following result shows that   among the  $\bbF$-$Q^*$-supermartingales contained in $ \cV(H_T)$ the respective  $\bbF$-$Q^*$-martingale yields the
   minimal possible price process, showing that the BN pricing formula in Equation \eqref{BPF} yields the cheapest strategy.

\begin{proposition}
	\label{BPprice}
Consider  a contingent claim $H_T \ge 0$ with the bounded stopping time $T>0$ as maturity    and assume that $\nicefrac{H_T}{S^*_T} \in L^1_{Q^*}(\ccF_T)$. If the fair price $H^*$ given by the BN pricing formula \eqref{BPF} can be realized by a trading strategy $\nu \in \cV(H_T)$, i.e.\ $H^*=V^\nu$, then $\nu$ is a BNRM strategy. 
\end{proposition}
\begin{proof}
For the convenience of the reader we repeat the short proof of Lemma A.1 in \cite{DuPl16}. 
	Consider the two non-negative  $\bbF$-$Q^*$-supermartingales $V^\nu$ and $V^{\tilde \nu}$ that satisfy
	$$ V_T^\nu = \frac{H_T}{S_T^*} = V_T^{\tilde \nu}. $$
	Since $V^\nu$ is an  $\bbF$-$Q^*$-martingale, $V_t^{\nu}=E^*[V_T^\nu|\ccF_t]$. Moreover, since $V^{\tilde \nu}$ is an  $\bbF$-$Q^*$-supermartingale we already obtain
	$$ V_t^{\nu}=E^*[V_T^\nu|\ccF_t]= E^*[V_T^{\tilde \nu}|\ccF_t] \le V_t^{\tilde \nu}$$
	$Q^*$-a.s.\ for all $0 \le t \le T$. Hence $\nu$ is the cheapest strategy and the claim follows.
\end{proof}

\subsection{Martingale representation} 
Up to now the results on the associated martingales were quite general. To obtain a more precise comparison it will be  useful to represent these martingales in terms of the driving processes. This extends  the F\"ollmer-Schweizer decomposition, as described in \cite{Schweizer95a}, to the BN setting. 
We introduce the  decomposition  in the following form:

\begin{definition}\label{HT}
	 The contingent claim $H_T$  is called {\em regular}  if it can be represented via a {\em BN decomposition}
	\begin{equation}\label{tildeHT}
\frac{	H_T}{S^*_T}=\frac{H_t^*}{S^*_t} +\int_{t}^{T}\zeta_s^\top \cdot dX_s +\eta_T - \eta_t
	\end{equation}
	where $\zeta \in L^2_{\loc,Q^*}(X)$, $\eta$ is a locally square-integrable  $\bbF$-$Q^*$-local martingale,  $\eta_0=0$ and $\eta$ is $Q^*$-orthogonal to $X$.
\end{definition}
Note the similarity to an admissible trading strategy that is mean-self financing and having a $Q^*$-orthogonal monitoring process. Intuitively, the above BN decomposition refers to the heuristic notion of the hedgeable part given by the stochastic integral and the completely nonhedgeable part represented by $\eta$. 
 Combining   Definition \ref{BPRISKMIN}, Corollary \ref{BPprice}, and the decomposition in Definition \ref{HT}   leads us directly to the following conclusion:
\begin{corollary}\label{cor:BNRM}
Consider  a regular contingent claim $H_T \ge 0$ with the bounded stopping time $T>0$ as maturity    and assume that $\nicefrac{H_T}{S^*_T} \in L^1_{Q^*}(\ccF_T)$. Then there exists a BNRM-strategy $\nu=(\eta,\zeta)$ that delivers $H_T$ and realizes the fair price process $H^*$ given by the BN pricing formula \eqref{BPF}. Here, $\zeta$ is given by Equation \eqref{tildeHT} and the monitoring process $\eta$ satisfies
\begin{align}\label{eq:eta}
	\eta = V^\nu - \int_0^\cdot \zeta_s^\top \cdot dX_s - V_0^\nu.
\end{align} 
\end{corollary}
Note that the BN decomposition in Equation \eqref{tildeHT} can be interpreted as a Kunita-Watanabe decomposition under $Q^*$. It  follows from \cite{AnselSt93} that when the  $\bbF$-$Q^*$-local martingale $X$ is  continuous,  all  contingent claims with $\nicefrac{H_T}{S^*_T} \in L^1_{Q^*}(\ccF_T)$ are regular. 
For a given regular contingent claim $H_T$, one can compute the BN decomposition  by first  calculating 
$$
H_t^*= S_t^* \, E^*\bigg[ \frac{H_T}{S_T^{*}} | \ccF_t \bigg] , $$ 
either by
explicit calculations or via  numerical methods. 
The position $\zeta$  in the self-financing part of the portfolio can  be obtained by application of the  It\^{o} formula. 

\subsection{An illustrative example}(Simplified market setting)
Let us provide an illustration of the above result for the previously introduced simplified market setting, which relates to Proposition 7.1 in \cite{DuPl16}. Our aim is to hedge a contingent claim and identify the hedgeable and nonhedgeable parts. To this end, consider a contingent claim $H_T$ and recall that $W^*$ given in Equation \eqref{eq:W*} is in the simplified market setting a one-dimensional $Q^*$-Brownian motion. 
We %
add an independent Brownian motion $W'$ to our setting, which is $Q^*$- orthogonal
   to $X$. Since  $W'$ does not drive the Radon-Nikodym density $\Lambda$, there is no measure change arising for $W'$. This means, $W'$ is a  Brownian motion under both $P$ and  $Q^*$. 

   In the simplified market setting $X$ represents with its two components the two primary security accounts $X^0$ and $X^1$ that are denominated in the stock GOP $S^*$ and driven by $W^*$: see Equation \eqref{eq:X} and Equation \eqref{e.4'''}. Assume that $X^1= 1$  and  
\begin{align}\label{ass:dX}
	dX^0_t =- X^0_t \,  \theta_t  dW^*_t, \qquad t \ge 0,
\end{align} 
with $\theta \in L^2_{\loc,Q^*}(W^*)$. 
We also assume that the representation is not redundant in the sense that $\theta_t  \neq 0$  for all $t \ge 0$. 

The contingent claim $H_T$ can  depend directly on $W'$ (for example through a dependence on the volatility $\theta$) and we assume that its fair price price process $H^*$, denominated in terms of the stock GOP (see Equation \eqref{BPF}), has the BN decomposition
	\begin{equation}\label{tildeHT''}
	\frac{H_t^*}{S_t^*} =\frac{ H_0^*}{S_0^*}+\int_0^t \varphi^*_s \cdot dW^*_s +\int_{0}^{t} \varphi'_s \cdot dW'_s,
	\end{equation}
	where $\varphi^* \in L^2_{\loc,Q^*}(W^*)$ and $\varphi'\in L^2_{\loc,Q^*}(W')$. Note that this implies that the contingent claim is regular. 

We compute the BNRM strategy $\nu=(\eta,\zeta)$ by Corollary \ref{cor:BNRM}:	first, by Equation \eqref{stochint}, we obtain for the self-financing part $X^\zeta$ that
$$ 
	dX^\zeta_t = \zeta_t \cdot dX_t =
	 -\zeta^0_t  X^0_t  \,  \theta_t  dW^*_t, 
$$	
where the second step follows by \eqref{ass:dX}. Uniqueness of the semimartingale decomposition yields that 
$$ \zeta^0_t X^0_t \theta_t =- \varphi_t^*, 
$$ 
$t \ge 0$. By the assumption that $\theta_t $ is not zero, 
we obtain 
$$ \zeta^0_t = -\varphi_t^* (\theta_t X^0_t)^{-1}.
$$
The condition \eqref{eq:eta}  yields 
$$ 
\eta_t = \int_0^t \varphi'_s \cdot dW'_s
$$
and it follows by \eqref{tildeHT''}
\begin{equation}
\zeta^1_t=\zeta^1_0+\frac{H^*_t}{S^*_t}-\frac{H^*_0}{S^*_0}-\eta_t.
\end{equation} This shows by \eqref{tildeHT''} that $\nu=(\eta,\zeta)$ is the BNRM strategy. Analogously, one can handle nonreplicable contingent claims in the  general setting.

\section{Working Capital and Refinancing}
\label{sec:working capital}

In this section we discuss how a financial institution or {\em line of business} (LOB) can produce the payoffs of derivatives and contracts  for its customers by employing the above BN pricing and hedging methodology. %
 The {\em production portfolio $P$} of an LOB is given by the initial position plus the value of the self-financing trading strategy aiming at reproducing the hedgeable parts of the contingent claims in the portfolio. The {\em working capital $C$}  is  needed by the LOB to produce the payoffs of the contingent claims in the book of the  LOB in a manner that avoids potential ruin and sufficiently contributes so that regulatory requirements are satisfied.

\subsection{Production Portfolio}
An LOB of an insurance company or a pension fund typically  holds numerous  liabilities in its book. We view  these liabilities as  contingent claims, representing payoffs that the LOB is obligated to produce over time and to deliver at the respective maturity dates to its clients.
In return for the delivery of the contingent claim, a client pays either a single premium at  initiation of the respective contract or makes periodic premium payments. Periodic premiums are studied in  \cite{bernard2017impact}, while we focus, for simplicity, on contracts with a single initial premium.\\ 
  The management of the LOB utilizes the paid premiums  in the production process to hedge the replicable  components of the contingent claims. Nonreplicable components are   diversified by managing these collectively in the book of the LOB.  %
  In the case of multiple LOBs, the remaining risks of the combination of these LOBs are further diversified in the overall book of the entire legal entity.\\ %
  
     Most   insurance and  pension  contracts involve contingent claims that are not fully replicable.  The intrinsic  randomness of these not fully replicable contingent claims impacts  the production of their payoffs. To manage this randomness and to avoid potential ruin, the LOB needs potential access to  additional external capital that may have to be used in the production process when necessary.   
         At the legal entity level, this should be made possible by taking into account a {\em capital cost provision} (called risk margin under Solvency II, or market value margin under the Swiss Solvency Test). This capital cost provision should, in particular, ensure that sufficient assets are set aside to provide compensation for new capital to be raised in the future %
   in the case of adverse risk developments.\\ %

In the following, we %
propose a transparent design of the  \emph{production process} for  a given set of contingent claims in the book of an  LOB. We  aim  at the most economical production method by  using BN risk-minimization and  use the stock GOP for  investing  any  extra capital kept in the production process.  Hence, this design of the production process is based on the,  in the long run, fastest growing, guaranteed strictly positive  investment, the stock GOP. %

To be more precise, consider a set of regular contingent claims $H^1_{T_1},\dots,H^m_{T_m}$ that are due at their respective maturities $T_1,\dots,T_m$. 
To the $i$-th contingent claim  we associate a self-financing trading strategy $\bar \zeta^i$ which aims at hedging its hedgeable part. 
The stock GOP-denominated dynamic production portfolio of the LOB is then given by the sum of these positions 
\begin{align}
	P := \sum_{i=1}^m \left( P_0^i + \int_0^\cdot  (\bar \zeta^i_s)^\top \cdot dX_s\right),
\end{align}
with initial position $P_0 = \sum_{i=1}^m P_0^i$, $i =1,\dots,m$.  After its maturity $T_i$, the production portfolio does no longer hedge the $i$-th claim, such that $\bar \zeta^i_t=0$ for $t>T^i$. We note that the trading strategy $\bar \zeta$ could possibly be different from the BNRM strategy $\zeta$, which we will introduce below.

\subsection{Working capital}
\label{ssec:Working capital}
The fair price $H^{i,*}$ of the $i$-th contingent claim $H^i_{T_i}$  is determined by the BN pricing formula \eqref{BPF},
$$ H^{i,*}_t = S_t^* \, E^*\bigg[ \frac{H^i_{T_i}}{S_{T_i}^*} \, \Big| \, \ccF_t \bigg], \qquad 0 \le t \le T_i. 
$$
The fair price does not include any compensation for the associated capital costs. In actuarial practice, the fair price is also called the \emph{best estimate}, which is  required, e.g., in the European Directive 2009/1388/CE Article 77 as BEL. The market-consistent value of the insurance obligations, however, includes an additional capital cost provision for compensating the risk capital providers. 

\begin{remark}[Additional capital costs]
\label{rem:Additional-capital-costs}
A capital cost provision can be added to our setting by increasing appropriately the initial premiums that the purchasers of contracts pay. However, this leads beyond the scope of the current paper.
\end{remark}

The capital that facilitates  the production of the payoffs of the contingent claims in the book of the liabilities of an LOB is called its \emph{working capital}. %
The working capital denominated  in terms of the stock GOP $S^*$ is denoted by the stochastic process $C'$. 
It  is given by the sum of the initial working capital $C'_0$ and the value of the dynamic production portfolio  $P$ minus the fair market value of the liabilities (all denominated in $S^*$), i.e. 
    \begin{equation}\label{eq:WC1}
     C'_t= C'_0+ P_t- \sum_{i=1}^{m} \frac{H^{i,*}_t}{S^*_t}, \qquad t \ge 0.
    \end{equation}
Since each contingent claim is assumed to be regular, Corollary \ref{cor:BNRM} allows us to introduce the  associated  BNRM strategies $\nu^i=(\eta^i,\zeta^i),$ $i=1,\dots,m$,  such that
$$ \frac{H^{i,*}_t}{S^*_t} = \frac{H^{i,*}_0}{S^*_0} + \int_0^t (\zeta^i_s)^\top \cdot dX_s + \eta^i_t. $$
Together with Equation \eqref{eq:WC1} we obtain the representation
\begin{align}\label{eq:C'}
	C'_t = C'_0 + P_0 -\sum_{i=1}^m \frac{H^{i,*}_0}{S^*_0} + \int_0^t \sum_{i=1}^m(\bar \zeta_s^i - \zeta_s^i)^\top \cdot dX_s - \sum_{i=1}^m \eta_t^i. 
\end{align}
For later reference we denote $H^* = \sum_{i=1}^m H^{i,*}$ and
$$ \eta := \sum_{i=1}^m \eta^i. $$
If the LOB follows the  BN risk-minimizing strategy starting with initial positions equal to the initial fair price, paid by the contractors of the contingent claims, it holds that  $$\label{P0i} P_0^i = \frac{H^{i,*}_0}{S^*_0},$$ $i\in\{1,...,m\}$, and we  obtain the working capital in the following form:
\begin{corollary}\label{cor:WC}
	If the trading strategy $\bar \zeta$ refers to the BNRM strategy $ \nu = (\eta,\zeta)$, i.e.\ $\zeta = \bar \zeta$ and \eqref{P0i},  then the stock GOP-denominated working capital equals
	\begin{align}\label{eq:WC}
		C' = C'_0 - \eta.
	\end{align}  
\end{corollary}

 \subsection{The refinancing process}\label{sec:counter-cyclical}
 There exist many possible ways of choosing the working capital  for the LOB  in  a way that may avoid ruin and sufficiently supports that the regulatory requirements are satisfied. From a prudent risk management point of view, the working capital should always stay above a \emph{critical capitalization level}.  This level has to be chosen appropriately such that the probability of ruin is sufficiently small, a cost efficient production method can be employed, and the regulatory requirements are satisfied.   For an LOB that applies BN risk-minimization, the working capital is invested in the stock GOP.  The latter is driven by randomness that is not diversifiable and difficult to avoid or offset in the production process. Therefore, the current paper suggests to employ a predictable, positive c\`adl\`ag process $D$ as critical capital level that is denominated in units of the stock GOP $S^*$. The latter may involve the consideration of a certain  \emph{safety margin}.

 \begin{remark}\label{rem:counter-cyclical}
 	Since  the critical capital level is in terms of $S^*$, the proposed BN risk-minimization  strategy  encourages \textquoteleft counter cyclical' risk management because in times of a market drawdown also the critical capital level $DS^*$ is dropping substantially. 
 \end{remark}

 One can argue that the stock GOP is, in a certain sense, in the long run the most secure and  fastest growing guaranteed strictly positive investment opportunity. It is arguably safer in comparison to the savings account of currencies that might be subject to devaluation or even default risk.
  When the working capital falls below the critical capital level, some action must be taken. This usually means that some \emph{refinancing} of the working capital in the LOB  must take place or bankruptcy is looming. It makes sense to implement a legally binding and transparent algorithm for such a refinancing strategy  where the management of the LOB is obliged to raise some prescribed external  capital when needed. To facilitate this, we introduce the c\`adl\`ag and piecewise constant \emph{refinancing capital process} $R$. The \emph{refinanced working capital } is consequently given by
 \begin{align}
 	C_t = C_t' + R_{t-}, \qquad t \ge 0. 
 \end{align}
 Note that in the refinanced working capital process we  omit the new refinancing capital at refinancing times. At initiation, one has $R_0=0$ and the production starts with sufficient capital, i.e.,
 \begin{equation}
 	C_0 = C'_0 > D_0 > 0.
 \end{equation}

 To specify the algorithm for a transparent refinancing process,  we introduce a refinancing time as a stopping time that is given by the first time after the previous refinancing time when the refinanced working capital falls below or reaches the critical capital level. 
 Accordingly,  
 we introduce the \emph{$k$-th refinancing time} 
 \begin{align}
 	\rho_k := \inf\{t > \rho_{k-1} \colon C_t \le D_t\},
 \end{align}
 for $k=1,2,\dots$ with $\rho_0=0$.  %
  At each refinancing time $\rho_k$ the refinanced working capital is raised to a sufficiently high level $\mu_k D_{\rho_k}$ which is determined by the  \emph{refinancing ratio} $\mu_k >1$. For illustration, we assume continuous working capital $C'$, constant critical capitalization level $D=D_0>0$, and constant refinancing ratio $\mu_k=\mu>1$ for $k\in\{1,2,...\}$. Consequently, the extra capital given by the difference between the envisaged new  level $\mu D$ and the capital available at the refinancing time $\rho_k$, $C_{\rho_k}=C'_{\rho_k}+R_{\rho_k-}=D$ is given by
 $ \mu D - D =D(\mu-1). $
This  determines the refinancing capital 
\begin{align}
\label{eq:R} R_t =(\mu-1) D \sum_{k \ge 1} \ind{\rho_k \le t}    . 	
 \end{align}
Refinancing  is provided from external funds that are raised in  the stock GOP $S^*$. 
The refinancing capital required at $\rho_k$  %
  assures that the refinanced working capital rises to the critical capitalization level. \\ 
 
 Formulating the refinancing strategy in the above transparent algorithmic manner offers a number of advantages: first, the underwriters can judge their  potential upcoming refinancing costs for a long period in advance. 
Second,    the supervisory authority can assess the financial position of the LOB at any time by using the available information about its assets and liabilities, its implemented refinancing strategy, and the financial market.
Third,  one could potentially determine optimal critical capitalization levels by following suggestions in \cite{LelandToft96}.

For a given  refinancing strategy, the resulting working capital $C$ follows a stochastic process that arises from the previously introduced BN risk-minimization and the exposure of the LOB to nonhedgeable randomness.  One can analyze the risk involved when running off scenarios for the refinanced working capital of  the LOB, e.g., by  calculating the value at risk, the expected short fall, or the ruin probability;  see, e.g.,  Section 8 in \cite{MFE}.\\

 Below we indicate how to obtain the BN price for the cost of refinancing over a fixed period. The refinancing times typically depend on the Brownian motions driving the primary security accounts. Therefore,   BN pricing with refinancing strategy requires modeling the refinancing times under the $Q^*$-probability measure.
For  illustration, assume that the working capital $C'$ %
follows a continuous strong Markov  process with known boundary crossing probabilities, e.g., a geometric Brownian motion, an Ornstein-Uhlenbeck process, or a squared Bessel process. %
 In this case  the unconditional $Q^*$-probabilities for the refinancing times can be calculated; see, e.g., Section 2.6.A in \cite{KaratzasShreve1988}, \cite{GoeingYor2003},  and  \cite{BorodinSalminen.02}. 	 For a time dependent function $D_t$ modeling the critical working capital level and the case of a geometric Brownian motion	 the methodology proposed in  \cite{SchmidtNovikov:2008} could be applied.
 This leads to the following  conclusion:
\begin{corollary}\label{lem:refinancing times}
Assume that the unconditional BN probability  for the $k$-th refinancing time arising before $T\in (0,\infty)$ is given by 
$ P^*(\rho_k \le T) $ for $k\in\{1,2,...\}$.	
	The BN price, denominated in the stock GOP $S^*$, of the refinancing cost for the period until time $T$  is given by the formula
	\begin{align}
	E^*[R_T] = (\mu-1)D \cdot \sum_{k \ge 1} P^*(\rho_k \le T). 
	\end{align}
\end{corollary}

\subsection{Diversification}
   Let us discuss the diversification effect of  a single LOB when the number $m$ of contingent claims in its book is sufficiently large and the associated monitoring processes are orthogonal. 
  	 Consider the BNRM trading strategies given in Corollary \ref{cor:BNRM}, where the monitoring processes $\eta_1,\dots,\eta_m$ are orthogonal (under $Q^*$) in the sense that $\langle \eta^i, \eta^k \rangle = 0$ for $i \neq j$.  
  The \emph{average working capital} per contract in the LOB is given by $\frac 1 m C$. The fluctuations of the working capital per contract can be measured via  its quadratic variation. Assuming  continuous monitoring processes,   the quadratic variation of the average working capital scales as in the weak law of large numbers. This leads to the following conclusion:
  \begin{corollary}
  	Assume that $\eta^1,\dots,\eta^m$ are continuous, orthogonal $\bbF$-$Q^*$-local martingales. Then the quadratic variation of the average working capital of the LOB is given by
  	\begin{align}\label{eq:average working capical}
  		 \langle  m^{-1}  C'  \rangle &=  \frac 1 {m^2} \sum_{i=1}^m \Big(
  		\int_0^\cdot \parallel \bar \zeta^i_s - \zeta^i_s \parallel^2 ds  + \langle \eta^i \rangle \Big). 
  	\end{align}
  \end{corollary}
  \begin{proof}
  	Recall Equation \eqref{eq:C'}, which yields that
  	\begin{align*}
  		\langle m^{-1}  C'  \rangle_t &= \frac 1 {m^2 }\sum_{i=1}^m 
  		\Big \langle  \int_0^t (\bar \zeta_s^i - \zeta_s^i)^\top \cdot dW^*_s - \eta_t^i \Big\rangle. 
  	\end{align*}
  	The result follows because each $\eta^i$ is by Definition \ref{HT} orthogonal to $W^*$ and by assumption $Q^*$-orthogonal to the other monitoring processes.
  \end{proof}
  In the case where the chosen trading strategy is the respective BNRM strategy, Equation \eqref{eq:average working capical} simplifies by Equation \eqref{eq:WC} to
  \begin{equation}
  	 \langle  m^{-1} C'  \rangle = \frac 1 {m^2}  \sum_{i=1}^m \langle \eta^i \rangle . 
  \end{equation}
  If the quadratic variations of the monitoring processes are bounded on the considered time interval $[0,\max\{T_1,\dots,T_m\}]$, and the time horizons remain bounded in the sense that $\max\{T_1,\dots,T_m\}<T^*<\infty$, then the quadratic variation of the average working capital converges to zero as the number of contingent claims, $m$, increases.   
  One can interpret this asymptotic property as the {\em diversification effect} that arises when pooling the production of the payoffs of a set of contingent claims in an LOB and applying  BNRM strategies. Note that one can derive even stronger results on the diversification effect as shown in  \cite{ArtznerEiseleSchmidt2024} using a conditional strong law of large numbers, as given in Theorem 3.5 in \cite{Majerek2005}. 
  
  \bigskip

   The success of an LOB  depends, in the long run, on the appropriate determination of the  initial working capital $C_0'$, the recapitalization ratio $\mu$, and the critical capitalization level.  Equation \eqref{eq:average working capical} highlights the fact that the  management of an  LOB has two main possibilities to reduce the necessary amount of average working capital  per contract: first, by hedging the replicable parts of the respective contingent claims as accurately as possible using BN risk-minimization, which leads to the minimal possible amount of the working capital needed for the hedgeable parts of the contingent claims; second, by pooling  more   contingent claims with preferably pairwise $ Q^*$-orthogonal stock GOP-denominated monitoring processes. The pooling increases for increasing $m$ the diversification effect, which means that it reduces the quadratic variation of the working capital proportional to $\frac{1}{m}$.   The critical capitalization level has to be determined by the management of the business line in such a way that, in particular, the regulatory requirements are satisfied.

For an overview on currently popular capital allocation methods in case of several LOBs we refer to Section 8.5 in \cite{MFE}. For a discussion of  economic properties of the  Euler principle for capital allocation (in particular, if applied based on a subadditive risk measure), we refer to \cite{tasche2008capital} and the literature cited therein. 
Since expected shortfall is highly sensitive to events in the far tail, in cases with extreme events, expected shortfall methods can lead to extremely high (and in turn for other LOBs low) capital allocations; see for example \cite{Pfeifer}. Finally, we refer to \cite{duran2022capital} for an analysis of different capital allocation methods under Solvency II, and to \cite{bielecki2020fair} for an account of unbiased capital allocation with improved backtesting properties. 

In the context of benchmark-neutral pricing and BN risk-minimization, it is important to note that current regulatory frameworks and modeling practice do rarely employ a proxy of the stock GOP $S^*$, like a stock index, as num\'eraire or reference unit. Instead, they typically employ the savings account as the num\'eraire and reference unit, cf.~e.g.~Section ii) in~\cite{BaFin2016}. As discussed in Section~\ref{sec:insurance-finance-arbitrages}, this reflects a particularly strong version of absence of arbitrage, which also excludes approximations of arbitrage that are hardly relevant in reality. Related to that, note that Article 22, Para. 3, lit b in the European Commission Delegated Regulation (EU) 2015/35 does not  require a risk-neutral valuation, but only the exclusion of arbitrage opportunities.  Note that in Section~\ref{sec:insurance-finance-arbitrages}, we show that benchmark-neutral pricing indeed effectively avoids insurance-finance arbitrages of the first kind, which corresponds to a wholly realistic no-arbitrage assumption, which is based on mathematical concepts that have become widely established. Furthermore, and as documented in~\cite{avo2023}  SR 901.011 (2023) Art.~28 of the changes (Aenderungen) of the AVO SR 901.011 (2023) and on page 21 in the respective details (Erlaeuterungen), also the regulation in Switzerland leaves some room for benchmark-neutral pricing since it is permitting mathematically founded valuation approaches with respect to different numéraires, which are, however, subject to approval. When applying benchmark-neutral pricing, care must be taken to ensure that this is consistent with the asset liability management used in practice and that the remaining risks are all consistently taken into account in the capital cost provision, cf.\ Section~\ref{ssec:Working capital} and in particular Remark~\ref{rem:Additional-capital-costs}. The choice of $S^*$ as numéraire also becomes crucial when determining the critical capitalization level for the working capital because this level is denominated in $S^*$. Recall, that according to Remark~\ref{rem:counter-cyclical} this yields a \textquoteleft counter cyclical' critical capital level in nominal terms. It is nevertheless well possible that the most suitable incorporation of BN-risk minimization into optimized regulatory capital requirement involves the appropriate consideration of several num\'eraires in the design of the capital requirement for an institution. However, this question goes beyond the scope of this paper and is left for future research.

\subsection{Non-replicable contingent claims in an  LOB}
 
For further illustration, consider an  LOB that produces the payoffs of  nonreplicable contingent claims, like non-traded wealth. The so-called QP-rule was  suggested in \cite{dybvig1992hedging} for that purpose and we refer to Section \ref{sec:QP} below for a discussion.  Assume, for simplicity, that the contingent claims in the book of the LOB have the same maturities, $T_1=\dots=T_m=:T\le T^*$, and that 
$H^i_{T}= \Ind_{A^i}$, i.e., the claims are binary options paying one unit of the savings account if $A_i$ occurs before $T$.  
If the events $A_1,\dots,A_m$ are independent of the driving Brownian motions and Assumption \ref{ass} holds, then Theorem \ref{thm:BNpricing} yields that
\begin{align}
	H_t^{i,*} &= S_t^* E^*\left[ \frac{\Ind_{A^i}}{S_T^*} \Big| \ccF_t \right] = 
	S_t^{**} E_P\left[ \frac{\Ind_{A^i}}{S_T^{**}} \Big| \ccF_t \right] = S_t^{**}E_P\left[ \frac{1}{S_T^{**}} \Big| \ccF_t \right] \cdot E_P[\Ind_{A^i} | \ccF_t] \notag %
	&=: P(t,T) \cdot P^i_t.
	\label{eq:decomposition non-replicable CC}
\end{align}
Here we denote the probability of the event $A^i$ by
$$ P^i_t := E_P[\Ind_{A^i} | \ccF_t] $$
and by 
 $P(t,T)$  the fair price of a zero-coupon bond that pays one unit of the savings account at maturity $T$. The latter is resulting from the BN-pricing formula \eqref{BPF}. Note that the independence assumption allows one to separate these two components. A classical change of num\'eraire could be considered when the change to the savings account as num\'eraire would be feasible with an equivalent measure change. Since this may not be realistic (see the following Section \ref{sec:MMM} for a discussion) we do not pursue such a classical measure change. %

By \eqref{BPF}%
, we obtain the  BN decomposition
\begin{equation}
\frac{H_T^{i,*}}{S^*_T} = \frac{H_t^{i,*}}{S^*_t} +\int_{t}^{T} P^i_s \, d \left(\frac{P(s,T)}{S^*_s}\right)+\int_{t}^{T}  \frac{P(s,T)}{S^*_s} \, dP^i_s
\end{equation}
because $P(.,T)$ is continuous and by the independence assumption the quadratic covariation vanishes. 
Since the zero-coupon bond price can be represented in terms of the driving process $W^*$, the  monitoring process is given by the integral
\begin{equation}
\eta_t^i = \int_{0}^{t} \frac{P(s,T)}{S^*_s}\, dP^i_s.
\end{equation}
The quadratic variation of the  working capital equals by  Equation \eqref{eq:WC}  
\begin{equation}
\langle m^{-1} C' \rangle_t 
= \frac 1 {m} \cdot \bigg( \frac 1 m 
\sum_{i=1}^m \int_{0}^{t} \left( \frac{P(s,T)}{S^*_s}\right)^2 \, d\langle P^i\rangle_s \bigg) ,
\end{equation}
which  decreases proportionally to $m^{-1}$ when the number $m$ of contingent claims in the book of the  LOB increases.

All working capital that is put aside to produce the payoffs of the contingent claims is invested in the stock GOP, which is, in the long run, the fastest-growing, guaranteed strictly positive  portfolio that the LOB can invest in. When including a capital cost provision in a generalization of the presented methodology, also the capital cost provision should be invested in the stock GOP.

If one would use  local risk-minimization under a putative minimal equivalent martingale measure, as described in \cite{FollmerSc91} and \cite{Moller98}, then the prices for the zero-coupon bonds we employ would equal  the value of one unit of the savings account. These prices are for long terms to maturity significantly higher than the fair prices that BN pricing identifies; see, e.g., \cite{Platen24}. Furthermore, under local risk-minimization  no dynamic hedging  would occur when the savings account denominated contingent claim is independent of the savings account denominated primary security accounts. The working capital that would be needed under local risk-minimization would be  invested in the savings account, which has shown over long time periods a significantly lower average growth rate than well-diversified stock portfolios, like the stock GOP. Therefore, local risk-minimization  exhibits the following two disadvantages: first, its prices are, in general, more expensive; and, second, it misses out on the  faster-growing investment of the  working capital  that BN risk-minimization  facilitates.

\section{Simplified stock GOP model}
\label{sec:MMM}

Following \cite{Platen24},  we consider in this section a parsimonious  model that captures  empirical properties of  the stock GOP; see  Section 13.1 in \cite{PlatenHe06} for a more detailed account. %
With reference to Section \ref{sec:simplified}, we call this model the \emph{ simplified stock GOP model}.

\subsection{The simplified stock GOP model dynamics}
The following model  captures three  well-established stylized empirical facts of diversified stock portfolios: the  {\em stationarity} of their volatilities; the \emph{leverage effect}, which links their volatilities to the ups and downs of the market; and the \emph{overall exponential growth} of diversified portfolios over long periods. These properties will be encapsulated in a time-changed model in the simplified market setting through a particular choice of its volatility.

To this end, consider the simplified market setting and, following Equation \eqref{eq:stock GOP dynamicsy}, rewrite the dynamics of $S^*$ in the form
\begin{align}\label{eq:6.1}
	 	\frac{dS^{*}_t}{S^{*}_t} &=\theta_t  ( \theta_t  + dB^*_t), \qquad t \ge 0
	 \end{align}
with $\theta = \, ||\, \sigma^* || $,
 where $B^*$ is the one-dimensional $Q^*$-Brownian motion given by 
$$ dB^*_t = \frac 1 {\theta_t} \,\sigma_t^*{}^\top dW_t^*, \qquad t \ge 0
$$ 
with $B^*_0=0$. The process $\theta$ is specified as
\begin{equation}\label{theta2'''}
\theta_{t}=\sqrt{\frac{4a e^{ \tau_t} }{S^*_t}}, \quad t \ge 0,
\end{equation}
which includes a time-transformation involving the \emph{activity time}
$$ \tau_t =\tau_0 + a t$$
with initial activity time $\tau_0$ and {\em activity} $a>0$. This introduces for the volatility $\theta$ its stationarity   through an appropriate choice of the activity time, the leverage effect through $S^*$  in the denominator of the squared volatility, and   the long-term exponential growth of the stock GOP through the exponential of the activity time in the drift and diffusion coefficients of the  SDE below. The choice of $\tau_0$  allows adjusting the model's volatility at the initial  time $0$. %
Inserting \eqref{theta2'''} into \eqref{eq:6.1} leads to the dynamics of the stock GOP  
\begin{align} 
\label{e.4''}
 dS^{*}_t
=4a e^{\tau_t}dt
+\sqrt{4a e^{\tau_t}} \, \sqrt{S_t^*} \,  dB^*_t, \qquad t \ge 0\end{align}
with $S^*_{0}>0$.
We thus arrive under $Q^*$ at a model in the simplified market setting, where
 $S^*$ is a time-transformed squared Bessel process of dimension four; see \cite{RevuzYo99}.

 The particular trajectory of the net risk-adjusted return $\lambda^*_t$, see \eqref{e.4'}, does not matter when applying BN pricing, which is extremely important from the practical perspective. For  simplicity, we assume that $\lambda^*=(\lambda^*_t)_{t \ge 0}$ 
 is a c\`adl\`ag, bounded deterministic function of the calendar time.

\begin{proposition}\label{prop:6.1}
	Consider the simplified stock GOP model given in Equation \eqref{e.4''} and
	a finite time horizon $T^*>0$. Define $Q^*$ by 
	\begin{align}
		\frac{dQ^*}{dP} = \Lambda_{T^*}
	\end{align}
	and assume  $\lambda^*=(\lambda^*_t)_{t \ge 0}$ 
	to represent a c\`adl\`ag, bounded deterministic function of the calendar time. %
	Then $Q^*$ is an equivalent local martingale measure when $S^*$ is chosen as the num\'eraire. Moreover, prices of $\ccF_{T^*}$-measurable contingent claims computed with the BN pricing formula \eqref{BPF} coincide with fair prices obtained by the real-world pricing formula in \eqref{eq:real-world pricing}.
\end{proposition}
\begin{proof}
First, note that the stochastic exponential $\Lambda$ satisfies, by Remark \ref{rem:3.3}, 
	$$ \frac{d\Lambda_t}{\Lambda_t} = - \frac{\lambda_t^*}{\theta_t} dB_t= - \frac{\lambda_t^*}{\sqrt{4a e^{\tau_t}}}  \, \sqrt{S_t^*} \, d B_t = - c(t) \sqrt{S_t^*} \,dB_t, \qquad t \ge 0 $$
	with the deterministic function $c(t) = \lambda^*_t \cdot (4a e^{\tau_t})^{-\nicefrac 1 2}$ . 
		Moreover, by Equation \eqref{e.4''} and Lemma \ref{sigma*minus} the stock GOP satisfies
	$$ dS_t^* =S^*_t\lambda^*_tdt+ a(t) dt + \sqrt{a(t)} \, \sqrt{S_t^*} \, dB_t, \qquad t \ge 0$$
	with the deterministic function $a(t) = 4 a e^{\tau_t}$. Its expectation equals
	\begin{equation}
E_P[S^*_t]=\Phi_{t}\left(E_P[S^*_0]+\int_{0}^{t}a(s)\Phi_s^{-1}ds\right)
	\end{equation}
		with 
		\begin{equation}
		\Phi_t=\exp\{\int_{0}^{t}\lambda^*_sds\}, \qquad t \ge 0.
		\end{equation} It is straightforward to show that the $\bbF$-$P$-local martingale $\Lambda$ is square-integrable. Using the above expectation of $S^*_t$, we obtain
		\begin{equation}
		E_P[\langle\Lambda_.\rangle_t]=\int_{0}^{t}c^2(s)E_P[S^*_s]ds=\int_{0}^{t}c^2(s)\Phi_s\left(E_P[S^*_0]+\int_{0}^{s}a(z)\Phi_z^{-1}dz\right)ds<\infty
		\end{equation}
	for $t\in[0,\infty)$. Since this expectation is finite, it follows by Corollary 3 on page 73 in \cite{Protter2005} that $\Lambda$   is a true  $\bbF$-$P$-martingale.	%

  In this case, the benchmark-neutral pricing measure is   an equivalent probability measure. 
 It  follows, as in Theorem \ref{thm:BNpricing}, that prices discounted with $S^*$ are $\bbF$-$Q^*$-local martingales and that benchmark-neutral pricing  provides the same fair prices as  real-world pricing.  
\end{proof}

\subsection{The relation to risk-neutral pricing} \label{sec:putative RN measure}

Classical  risk-neutral pricing employs the  riskless savings account  as num\'eraire for  pricing and relies on the  no free lunch with vanishing risk (NFLVR) no-arbitrage condition, which is equivalent to the existence of an equivalent risk-neutral probability measure, which we denote by $Q^0$.  Under $Q^0$  prices denominated in the savings account are $\bbF$-$Q^0$-local martingales. When the benchmark-neutral pricing measure $Q^*$ is an equivalent probability measure, it has been shown by Theorem 4.2 in \cite{Platen24} that the  putative risk-neutral pricing measure $Q^0$ is not an equivalent probability measure.
This stems from the fact that the Radon-Nikodym density changing from $Q^*$ to $Q^0$ is given by
$$ \Lambda_T^0 = \frac{1}{S_T^*}$$ 
(recall that $X$ is already stock GOP-denominated). As already remarked,  $S^*$ is under $Q^*$ a time-transformed squared Bessel process of dimension four, such that its inverse $\Lambda^0$ becomes  a strict local martingale and $Q^0$, therefore, not equivalent to $Q^*$ (and hence to $P$).
The consequence for our setting is that  risk-neutral prices for many contingent claims are more expensive  than necessary.

Furthermore, we remark that by choosing the  savings account as the num\'eraire, one  invests the working capital in the savings account; see, e.g., \cite{Moller98}. However, in the long run, this investment does    not provide the fastest-growing, guaranteed strictly positive production portfolio. Therefore, when designing the risk management around the savings account, one is missing out  on the faster growing production portfolios that are possible because they involve the stock GOP instead of the savings  account for investing their reserves in the num\'eraire. The latter num\'eraire   explains intuitively through its faster average long-term growth the difference between the BN prices and  respective risk-neutral prices. 
This difference  relies on the assumption in Equation \eqref{theta2'''}. If  different dynamics of $\theta$ would be more appropriate in real markets, risk-neutral prices and BN prices could  coincide; see for example Section 6 in \cite{FilipovicPl09}.

\section{Insurance-finance arbitrages}		\label{sec:insurance-finance-arbitrages}
		
For application to the pricing of insurance contingent claims that  are linked   to the financial market, BN pricing and hedging is ideally suited. We show in the following that BN risk-minimization results in a price system that avoids so-called insurance-finance arbitrage, to be made precise below. This line of research was initiated in \cite{ArtznerEiseleSchmidt2024}.  
Moreover, BN risk-minimization allows us to extend the approach proposed in \cite{ArtznerEiseleSchmidt2024} without assuming the existence of an equivalent risk-neutral probability measure (and, hence, beyond the NFLVR condition), as we will show in the following. %

As before, $X=(X^0,X^1,\dots,X^n)^\top$ denotes the vector of primary security accounts  denominated in units of the  stock GOP $S^*$. For simplicity, we assume $S^*_0=S^{**}_0=1$. 
As in Section \ref{sec:Benchmark-neutral risk minimization}, we  consider in this section a self-financing trading strategy $\nu=(0,\zeta)$ that invests according to some process  $\zeta\in L^2_{\loc,Q^*}(X)$ in the primary security accounts, starting from an initial stock GOP denominated  value $\zeta_0^\top X_0 \in \R$; compare Equation \eqref{def:admissible-supermartingale}. The accumulated wealth until time $t$ is given by
$$ 
	V_t^{(0,\zeta)} = \zeta_0^\top   X_0 + \int_0^t \zeta_s \, dX_s. 
$$
For the following we fix a finite future horizon $T>0$ and call   $\zeta$ \emph{permissible} if $V_t^{(0,\zeta)} \ge 0$ for  all $0 \le t \le T<\infty$.

As insurance contract we consider a contract which  offers a variety of benefits at a future time in exchange for a single premium. We work with savings account denominated quantities and, without loss of generality, we consider  a single premium paid at time $t=0$ and all benefits received at the future time $T$. It is important to point out that insurance contracts are typically not $\bbF$-adapted, since insurance payoffs depend on the personal information of the insured. For simplicity, we consider only two time points $t=0$ and $T>0$ and describe the insurance information by $\bbG_0=\ccF_0$ and $\ccG_T$, which is possibly larger than $\ccF_T$. On the other hand, it is  natural that
$$ \ccF_T \subset \ccG_T. $$
For simplicity, we assume additionally that $\ccG_T \subset \ccF$, where we recall that we work on the probability space $(\Omega,\ccF,P)$. Alternatively, one could also extend the probability space, where special care has to be taken not to introduce forms of arbitrage.

Regarding the $i$-th insurance contract,  we denote by $p^i \in \mathbb{R}$  the savings account-denominated premium to be paid at time $0$ and
 a $\ccG_T$-measurable savings account-denominated benefit $\beta^i \ge 0$   to be received by the owner of the $i$th contract at time $T$, $i \ge 1$. This allows us to cover a wide range of contracts, including contracts linked to the financial market, such as  variable annuities, where we refer to \cite{BallottaEberleinSchmidtZeineddine2020, ballotta2021fourier} and references therein. While the mentioned papers focus on L\'evy processes, it is reasonable to expect that these settings can also be extended to cover general affine processes, as studied, for example, in \cite{KellerResselSchmidtWardenga2018}.

  The weakest choice of the conditioning $\sigma$-algebra in this setting is given by 
$$ \ccH := \ccG_0 \vee \ccF_T. $$
This information includes the full evolution of the financial market until time $T$ and the initial insurance information. On  one side this is a large $\sigma$-algebra, which means that the above is a rather weak assumption. On the other side, including $\ccG_T$ would imply that the insurer has access to the future information $\ccG_T$ at time $0$. Hence, $\ccH$ is also the largest possible choice.

 We consider  a homogeneous pool of standard insurance payoffs ${\beta^1},\beta^2,... $ with the following properties:
 \begin{assumption} \label{assum:ConditionalExpectations}
 	Assume that \begin{enumerate}[(i)]
 		\item $\beta^1,\beta^2,... $ are $\ccH$-conditionally independent,
 		\item $E_P[\beta^i\vert \ccH]= E_P[\beta^1\vert \ccH]$ for all $i \ge 1$,
 		\item $\textnormal{Var}_P[\beta^i \vert \ccH]=\textnormal{Var}_P[\beta^1 \vert \ccH]$ for all $i \ge 1$.
 	\end{enumerate}
 \end{assumption}

	We model a typical traditional insurance business where an insurance company holds the insurance contracts until maturity and does not hedge these actively.  This is more restrictive than an insurance business that applies BN risk-minimization, where hedging  is actively pursued.
An insurance company forms an \emph{allocation} (at time $0$) of $m\in\{1.2....\}$ contracts from the homogeneous pool of payoffs. \footnote{Contracts at different points in time are not treated here for simplicity - this simplifies the treatment significantly; see \cite{OberprillerRitterSchmidt2024}.}.  More precisely,  an allocation  $\psi=(\psi^i)_{1 \le i \le m}$ is a non-negative, $\ccG_0$-measurable vector. Here, $\psi^i \ge 0$  denotes the size of the $i$-th  insured payoff. 

At time $0$ the total savings account-denominated premium  $p=\sum_{i=1}^{m}p^i$   is received by the insurance company for the sum of the $m$ insured benefits. If this premium would be invested  totally in the stock GOP $S^{*}_t$, then it would generate the savings account-denominated payoff $p S_T^{*}$  at time $T$. %
Hence, for this strategy by the insurer, the savings account-denominated value $U^\psi_T$ of the allocation $\psi$ at time $T$ is given  by the sum of the paid premiums minus the associated benefits, i.e. 
$$
U^{\psi}_T:= \sum_{i = 1}^m \psi^i (p^iS^*_T\,-\beta^i).
$$

In the following we  study in the given insurance-finance market  arbitrages of the first kind where the insurance company can dynamically reallocate or hedge securities by trading in the financial market. We follow  \cite{kardaras2012market} and refer to \cite{karatzas2021portfolio} for a detailed exposition on the topic. Note that arbitrages of the first kind (A${}_1$s) are a weaker concept in comparison to  free lunches with vanishing risk (FLVRs). The situation we consider here, with existing  GOP $S^{**}$ of the extended financial market, might not guarantee NFLVR but guarantees no unbounded profit with bounded risk (NUPBR), as shown in \cite{karatzas2021portfolio}. %

\begin{definition}
	An \emph{insurance-finance arbitrage of the first kind} (IFA${}_1$) is an $\ccF_T$-measurable random variable $\varepsilon$ satisfying $P(\varepsilon \geq 0) = 1$ and $P(\varepsilon > 0) > 0$, for which there exists, for every $x > 0$, an insurance allocation $\psi$ and  a permissible self-financing trading strategy $\zeta$ satisfying	$x=\zeta_0^\top X_0$ and
	$$ U_T^{\psi} + V_T^{(0,\zeta)}
	S^*_T 
	\ge \varepsilon. $$
\end{definition}

\smallskip

If no such arbitrage exists, we shall say that NIFA${}_1$ holds. Note that this condition requires NA${}_1$, since $\psi\equiv {\bf {0}}$ is a permissible insurance strategy. It, however, does not require NFLVR. We now provide the simple direction of the fundamental theorem of asset pricing on the absence of IFA${}_1$ in our setting. 

\begin{proposition} \label{prop:FTIFA}
Assume that Assumption \ref{ass} and Assumption \ref{assum:ConditionalExpectations} hold. If the $i$-th insurance premium $p^i=p$ is bounded by the BN-price, i.e.~if 
\begin{align}\label{eq:QP BN}
	p \le E_P\bigg[ \frac{\beta^i}{S_T^{**}} \, \Big| \, \ccG_0\bigg] = 
 E^*\bigg[ \frac{\beta^i}{S_T^*} \, \Big| \, \ccG_0\bigg] 
\end{align}
for all $i\in\{1,...,m\}$, then NIFA${}_1$ holds.
\end{proposition}
\begin{proof}
To begin with, we note that by Equation \eqref{eq:H H*}
\begin{align}\label{temp598}
E^*\bigg[ \frac{\beta^i}{S_T^*} \, \Big|\, \ccG_0\bigg] = E^*\bigg[ \frac{\beta^i}{S_T^*} \, \Big| \, \ccF_0\bigg] =
E_P\bigg[ \frac{S^*_T}{S_T^{**}}\frac{\beta^i}{S_T^{*}}  \, \Big| \,  \ccF_0\bigg] =
E_P\bigg[ \frac{\beta^i}{S_T^{**}}  \, \Big| \,  \ccF_0\bigg]
\end{align}
for all $i\in\{1,...,m\}$ because $\ccG_0=\ccF_0$. 
Furthermore, 
	note that by our previous assumptions there exists a GOP $S^{**}$ in the financial market including the savings account $S^0$. By Proposition 2.41 in \cite{karatzas2021portfolio}, the inverse of the GOP $S^{**}$ is a supermartingale deflator, and by Proposition 2.4 therein, this is equivalent to being a local martingale deflator, i.e., primary security accounts denominated by $S^{**}$ are $\bbF$-$P$-local martingales. 
	
Now, we	 choose a sequence of permissible wealth processes $V^{(0,\zeta_k)}$ with $x_k=\zeta^\top_0 X_0$ and insurance allocations $\psi_k$ such that $x_k \to 0$ and 
	 \begin{align}\label{temp:607}
	U_T^{\psi_k} + V_T^{(0,\zeta_k)}S^*_T \ge \varepsilon
	 \end{align}
 for some  non-negative random variable $\varepsilon \ge 0$. Denote by $Z=(S^{**})^{-1}$ the local martingale deflator associated to $S^{**}$. 
 
 For the insurance part, note that for any allocation $\psi_k$ one has
 \begin{align}
 	E_P\Big[ Z_T  U^{\psi_k}_T  \Big] &=
 	E_P\Big[ Z_T  \sum_{i = 1 }^m \psi^i_k (p\, S_T^{*}-\beta^i) \Big] \notag \\
 	&= \sum_{i = 1 }^m p \psi^i_k E_P\Big[  \frac{S^*_T}{S^{**}_T}  \Big] - 
 	E_P\Big[ Z_T \sum_{i = 1}^m \psi^i_k \beta^i \Big].  \label{AES1}
 \end{align}
 For the following, we introduce $\gamma_k=\sum_{i=1}^{m} \psi^i_k$. Moreover, since $Z_T$ and $\psi^1_k,\dots,\psi^m_k $ are $\ccH$-measurable, we obtain by Assumption \ref{assum:ConditionalExpectations}
 \begin{align}
 	E_P\Big[ Z_T \sum_{i = 1}^m \psi^i_k \beta^i \Big] &= E_P\Big[ Z_T \sum_{i = 1}^m \psi^i_kE_P \big[\beta^i | \ccH \big] \Big] 
 	= E_P\Big[ Z_T \sum_{i = 1}^n  \psi_k^iE_P \big[\beta^1 | \ccH \big] \Big] \notag\\
 	&= E_P\big[ Z_T \, \gamma_k \, \beta^1  \big]. \label{AES2}
 \end{align}
 
 Then,
	 by Equations \eqref{AES1} and \eqref{AES2},
	 \begin{align}
	 	E_P\Big[ Z_T \big( U_T^{\psi_k} + V_T^{(0,\zeta_k)}S^*_T \big) \Big]
	 	= E_P\big[  \gamma_k (p -Z_T\beta^1)\big] + 
	 	E_P\big[ Z_TS^*_T V_T^{(0,\zeta_k)} \big].
	 \end{align}
	 By condition \eqref{eq:QP BN},
	 \begin{align}
	 	E_P\big[  \gamma_k (p -Z_T\beta^1)\big] &= E_P\Big[\gamma_k \Big(p -\frac{\beta^1}{S_T^{**}}\Big) \Big] \le 0,
	 \end{align}
	 where we used \eqref{temp598}. Moreover, since $ZS^* V^{(0,\zeta_k)}$ is a nonnegative local martingale, and hence a supermartingale,
	 \begin{align*}
	 	E_P\big[ Z_TS^*_T V_T^{(0,\zeta_k)} \big] & \le  V_0^{(0,\zeta_k)} = x_k  
	 \end{align*}
	 for all $k \ge 1$. Hence, 
	 \begin{align*}
	 	E_P\Big[ Z_T \big( U_T^{\psi_k} + V_T^{(0,\zeta_k)}S^*_T \big) \Big] \le 0
	 \end{align*}
	 and, therefore, by Equation \eqref{temp:607},
	 \begin{align*}
	 0 \ge 	E_P\Big[ Z_T \big(U_T^{\psi_k} + V_T^{(0,\zeta_k)} S^*_T\big) \Big] \ge E_P[Z_T \varepsilon].
	 \end{align*}
	 Since $\varepsilon \ge  0$ this implies $\varepsilon = 0$ and, therefore, $\varepsilon$ cannot be an arbitrage of the first kind, such that NIFA${}_1$ holds. 
\end{proof}

\begin{remark}In comparison to \cite{ArtznerEiseleSchmidt2024}, and in particular to \cite{OberprillerRitterSchmidt2024}, we consider here insurance strategies with only finitely many contracts. For the shown result in Proposition \ref{prop:FTIFA}, this is the easiest approach and suffices our purpose. Extending this to asymptotic insurance-finance arbitrages of the first kind requires substantial more work, since Proposition B.1 of \cite{OberprillerRitterSchmidt2024} is no longer applicable in the setting with local martingale deflators. 	
\end{remark}

\subsection{The relation to the QP-rule}
\label{sec:QP}
The QP-rule, originally proposed in \cite{plachky1984conservation} and used in \cite{dybvig1992hedging} for the pricing of non-traded wealth, was applied to the insurance setting in \cite{ArtznerEiseleSchmidt2024}. It essentially combines the putative risk-neutral measure with the real-world probability measure in a way that does not allow asymptotic insurance-finance arbitrages. 

More precisely,  we consider the filtered probability space $(\Omega,\ccF_T,(\ccF_t)_{0 \le t \le T},P)$ for the financial market and the probability space $(\Omega,\ccG_T,P)$ for the insurance market where we recall that $\ccF_T \subset \ccG_T$. If we assume that NFLVR holds, then there exists an equivalent risk-neutral local martingale measure $Q$. The QP-rule specifies a unique measure $\QcirP$ which guarantees market consistency on the one side and valuation of insurance claims by conditional expectations on the other side. Indeed, it is the unique measure such that for a $\ccG_T$-measurable random variable $\beta \ge 0$, $E_{\QcirP}[\beta|\ccF_t] = E_Q[ E_P[\beta|\ccF_T] | \ccF_t]$. Asymptotic insurance-finance arbitrage-free pricing can now be performed by computing the $\QcirP$-conditional expectation of the discounted benefits. 
 
For the following conclusion, let us denote the density of $Q$ with respect to $P$, on $(\Omega,\ccF_T)$, by $Z_T=dQ / dP$.
\begin{proposition}
	If Assumption \ref{ass} and NFLVR hold with $S^0_0=S^{**}_0$, there exists an equivalent risk-neutral local martingale measure $Q$, such that the QP-rule coincides with the BN-price, i.e.
	\begin{align*}
		E^*\bigg[ \frac{\beta}{S_T^*} \,\Big|\, \ccG_0\bigg] =	E_P\bigg[ \frac{\beta}{S_T^{**}} \,\Big|\, \ccG_0\bigg] = E_{\QcirP} \bigg[ \frac{\beta}{S^0_T} \,\Big|\, \ccF_0 \bigg] = E_{\QcirP} \big[ \beta \,\Big|\, \ccF_0 \big]
	\end{align*}
	for any $\ccG_T$-measurable benefit $\beta \ge 0$. 
\end{proposition}

In reality, it seems unlikely that the assumptions of the above result are satisfied. However,
one can conclude the BNP-rule, which would hold under NUPBR. Most generally, one has  the $P$-rule of  the benchmark approach, which employs just the  real-world probability measure $P$, which  requests NUPBR. The latter is equivalent to the existence of the GOP $S^{**}$. Under Assumption \ref{ass}, however, the $P$-rule is equivalent to the BNP-rule.

  	\section*{Conclusion}

  This paper introduces the method of  benchmark-neutral risk-minimization for the pricing and hedging of not fully replicable contingent claims. This approach offers several advantages over local risk-minimization. Specifically, benchmark-neutral risk-minimization does not depend on the existence of either the minimal martingale measure or an equivalent risk-neutral local martingale measure, thereby allowing for a broader and more flexible modeling framework.
The method is based on the assumption of the existence of a growth optimal portfolio (GOP) for the entire market, as well as a stock GOP for the market  without a savings account. Additionally, it assumes that the stock GOP, denominated in units of the overall GOP, forms a martingale -- an assumption that appears to be realistic in  practical contexts.
Using the stock GOP as the num\'eraire, it is shown that the prices derived under the equivalent BN pricing measure align with the minimal possible prices. In competitive markets, this should decrease the prices of insurance contracts. Furthermore,  a transparent risk management framework is introduced for the least expensive production of contingent claims within a line of business. This framework explicitly describes a refinancing process where the reserve capital is allocated to the stock GOP instead to the savings account, thereby minimizing in the long run production costs.
An important feature of  benchmark-neutral risk-minimization  is that it allows for the total stock GOP-denominated profit and loss of a sufficiently diversified portfolio of contingent claims to asymptotically approach zero as the number of contingent claims increases. This result underscores the potential for effective risk reduction through diversification.
Moreover,   the feasibility of the approach is illustrated by considering a simplified stock GOP model that incorporates a range of stylized empirical facts in its dynamics. This highly tractable model demonstrates the practical applicability and robustness of the BN  risk-minimization methodology.
Finally,  an insurance-finance market is introduced and associated insurance-finance arbitrage of the first kind studied. In this setting, it is shown that any price  equal to or below  the BN price guarantees absence of insurance-finance arbitrages of the first kind.

 \end{document}